\documentclass[11pt,a4paper]{article}
\usepackage{graphicx,color}
\usepackage{amsmath,amssymb,dsfont,fullpage,epsfig,multirow,longtable,bm}
\usepackage{epsfig,epstopdf,hhline}
\usepackage{cases}
\usepackage{wrapfig}
\usepackage{algorithm}
\usepackage{algorithmic}
\usepackage[top=1 in, bottom=1 in, left=1 in, right=1 in, letterpaper]{geometry}
\usepackage{tikz}
\usetikzlibrary{arrows,positioning,decorations.pathreplacing,shapes}

\usepackage{thmtools}
\usepackage{thm-restate}
\usepackage{enumerate}

\usepackage{mdframed}
\usepackage{setspace}
\usepackage{caption}
\usepackage{xfrac}

\usepackage{authblk}

\newenvironment{proof}{\noindent\emph{Proof\ }}{\hspace*{\fill}$\Box$\medskip}

\newtheorem{theorem}{Theorem}

\newtheorem{lemma}{Lemma}

\newtheorem{proposition}{Proposition}

\usepackage{amsmath}
\usepackage{paralist}
\usepackage{framed}

\newcommand\restr[2]{{
  \left.\kern-\nulldelimiterspace 
  #1 
  \vphantom{\big|} 
  \right|_{#2} 
  }}

\newcommand{\vect}[1]{\ensuremath{\bm{#1}}}

\newcommand{\LW}{\textsc{Lw}}

\newcommand{\PoA}{\textsc{PoA}}

\usepackage[utf8]{inputenc} 
\usepackage[T5]{fontenc} 

\usepackage{color, colortbl}
\usepackage{hyperref}
  \definecolor{mydarkblue}{rgb}{0,0.08,0.45}
  \hypersetup{ %
    colorlinks=true,
    linkcolor=mydarkblue,
    citecolor=mydarkblue,
    filecolor=mydarkblue,
    urlcolor=mydarkblue,
    }

\usepackage[capitalise,noabbrev]{cleveref}

\usepackage[numbers]{natbib}

\begin{document}

\title{Efficiency of Proportional Mechanisms in Online Auto-Bidding Advertising}

\author{Nguyễn Kim Thắng}
\affil{LIG, Grenoble INP, University Grenoble-Alpes, France}

\maketitle

\begin{abstract}
The rise of automated bidding strategies in online advertising presents new challenges in designing and analyzing efficient auction mechanisms. In this paper, we focus on proportional mechanisms within the context of auto-bidding and study the efficiency of pure Nash equilibria, specifically the price of anarchy (PoA), under the liquid welfare objective. We first establish a tight PoA bound of 2 for the standard proportional mechanism. Next, we introduce a modified version with an alternative payment scheme that achieves a PoA bound
of $1 + \frac{O(1)}{n-1}$
where $n \geq 2$ denotes the number of bidding agents. This improvement surpasses the existing PoA barrier of 2 and approaches full efficiency as the number of agents increases. Our methodology leverages duality and the Karush-Kuhn-Tucker (KKT) conditions from linear and convex programming. 
Due to its conceptual simplicity, our approach may offer broader applications for establishing PoA bounds.
\end{abstract}


\section{Introduction}		\label{sec:intro}
%



The online advertising ecosystem has been rapidly evolving with the widespread adoption of auto-bidding systems by advertisers and markets \cite{AggarwalBadanidiyuru24:Auto-biddingandauctions}. Traditionally, advertising platforms required advertisers to manually submit detailed bids. However, with advancements in AI technologies, intelligent models can now serve as virtual consultants—assisting advertisers with analysis, strategy formulation, and decision-making. Specifically, an auto-bidding agent can translate an advertiser’s goals and constraints into strategic bids to optimize their objectives. This automation simplifies the advertiser’s role, enabling them to delegate complex tasks to the agent. At the same time, the growing presence of numerous optimized auto-bidding agents, each pursuing different goals, increases the complexity of the overall system. This introduces significant research challenges, particularly in understanding system equilibria and designing effective mechanisms. As a result, auto-bidding has become a prominent focus of recent research \cite{AggarwalBadanidiyuru19:Autobiddingwith,AggarwalBadanidiyuru24:Auto-biddingandauctions,LiawMehta23:Efficiencyofnon-truthful,Mehta22:Auctiondesign,DengMao24:Efficiencyofthefirst-price} within the advertising domain.


\subsection{Model and Definitions}

We define the problem faced by the autobidding agents and the auctioneer in an ad-auction environment. 
There are $n$ autobidding agents (indexed by $i$) and $m$ divisible items (indexed by $j$), and each item is sold in an independent auction.  
Each agent has a private function $v_{i}: [0,1]^{m} \rightarrow \mathbb{R}_{\geq 0}$ which defines its valuation 
over fractional allocations of items and a private budget $W_{i}$ representing the maximum total amount agent $i$ is willing to pay. 
We assume that functions $v_{i}$'s are continuous, differentiable, concave, and non-decreasing on every component. 
For each item $j$, each agent $i$ submits a bid $b_{ij} \in \mathbb{R}_{\geq 0}$.
After collecting all bids $b_{1j}, \ldots, b_{nj}$ on item $j$, the auctioneer determines the allocation $d_{ij} \in [0,1]$ such that $\sum_{i=1}^{n} d_{ij} = 1$ 
and the corresponding payment $p_{ij} \geq 0$ for each agent $i$. Let $d_{i} = (d_{i1}, \ldots, d_{i,m})$ denote the item fractions received by agent $i$. 
The valuation of agent $i$ in this allocation 
is $v_{i}(d_{i})$. In the context of online advertising, the valuation $v_{i}(d_{i})$ is typically $v_{i}(d_{i}) = \sum_{j=1}^{m} v_{ij} d_{ij}$ where 
$v_{ij} \in \mathbb{R}_{\geq 0}$ represents the value agent $i$ derives from fully acquiring item $j$.
In this work, we consider more general forms of valuation functions to capture a broader range of advertiser preferences.

\subparagraph{Agents' objectives and constraints.} 
In autobidding, agents have different objectives subject to their own constraints. The following are widely used objectives and their generalization.
\begin{itemize}
	\item Utility-maximizing objective: $v_{i}(d_{i}) - \sum_{j=1}^{m} p_{ij}$. 
	\item Valuation-maximizing objective: $v_{i}(d_{i})$. 
	\item Hybrid objective:  $v_{i}(d_{i}) - \rho_{i} \sum_{j=1}^{m} p_{ij}$ where $0 \leq \rho_{i} \leq 1$ is a parameter specific to agent $i$.
	This objective generalizes the ones above that correspond to $\rho_{i} = 1$ and $\rho_{i} = 0$. 
\end{itemize}
In the context of online advertising, utility-maximizing agents aim to optimize the difference between their valuation and the payment made. This objective is standard in auction and mechanism design theory. However, it can be challenging for advertisers to quantify their valuation in purely monetary terms \cite{AggarwalBadanidiyuru24:Auto-biddingandauctions}. In contrast, valuation-maximizing agents focus on directly maximizing outcomes such as clicks or conversions—captured by their valuation functions—while considering payments only indirectly, typically through budget or other constraints (as discussed below). Finally, hybrid agents pursue a combination of both goals, balancing the desire for high valuation outcomes and the payment incurred.

The most common constraints for agents $1 \leq i \leq n$ are the budget and the return-on-spend constraints. 
\begin{itemize}
	\item Budget constraint (\textsc{Budget}): $\sum_{j=1}^{m} p_{ij} \leq W_{i}$
	\item Return-on-spend  constraint (\textsc{RoS}): $v_{i}(d_{i}) \geq \tau_{i} \cdot \sum_{j=1}^{m} p_{ij}$. By \cite{DengMao24:Efficiencyofthefirst-price}, without loss of generality, we can assume $\tau_{i} = 1$ for all $i$.
\end{itemize}
The budget constraint is natural for agents to control their expenses in an ad campaign. 
 In addition, the \textsc{RoS} constraint captures a variety of considerations,  such as cost-per-action constraint, return-on-investment constraint, individual rationality constraint, etc.
It is particularly useful when certain ad campaign goals cannot converted into explicit monetary terms.      

To summarize, the problem of a bidding agent $i$ can be formulated as the following (convex) program:
$$
\max  v_{i}\bigl( d_{i} \bigr) - \rho_{i} \sum_{j=1}^{m} p_{ij} 
\quad 
\text{subject to:}
\quad
  \sum_{j=1}^{m} p_{ij} \leq W_{i} ~ (\textsc{Budget}),  
\quad
\text{and} 
\quad
 v_{i}\bigl( d_{i} \bigr)  \geq  \sum_{j=1}^{m} p_{ij} ~ (\textsc{RoS})
$$

%

\subparagraph{Auction objective.} 
On the side of the auctioneer, the objective is to maximize the \emph{liquid welfare}, 
defined as the total valuation of each agent capped by its budget, i.e.,  
$\LW(\vect{b};\vect{v}, \vect{W}) = \sum_{i=1}^{n} \min \bigl \{ W_{i}, v_{i}(d_{i})\bigr\}$ where 
$\vect{v} = (v_{1}, \ldots, v_{n})$ and $\vect{W} = (W_{1}, \ldots, W_{n})$ are the valuation and budget profiles
of agents. It has been observed that liquid welfare 
is more suitable than the classic social welfare (i.e., $\sum_{i=1}^{n} v_{i}(d_{i})$) since the latter 
cannot be well-approximated when agents are constrained. Moreover, liquid welfare represents 
the highest possible revenue that can be attained with full information on the agents' information.

\subparagraph{Equilibrium concepts and the PoA.}
In this paper, we focus on pure Nash equilibria. 
A \emph{pure Nash equilibrium} is a bidding profile $\vect{b}$ such that no agent can increase 
its objective via a unilateral deviation; that is, for every agent $i$ and every strategy $b'_{i} = (b'_{i1}, \ldots, b'_{i,m})$, 
$
u_{i}(\vect{b}) \geq u_{i}(b'_{i},\vect{b}_{-i})
$
where $\vect{b}_{-i}$ denotes the bids chosen by all agents other than $i$ in $\vect{b}$ and $u_{i}$ denote the objective of agent $i$
(which can be its utility, valuation, or hybrid objective with its own parameter $\rho_{i}$ depending on the agent). 
To quantify the efficiency of an auction, we consider the standard price of anarchy. 
The \emph{price of anarchy (PoA)} is defined as the largest ratio between the liquid welfare of an optimal solution and that of an equilibrium.  
Specifically, the price of anarchy of pure equilibria is defined as
$
\frac{\textsc{Opt}(\vect{v}, \vect{W})}{\LW(\vect{b};\vect{v}, \vect{W})}
$
where
$\textsc{Opt}(\vect{v}, \vect{W})$ is the optimal liquid welfare with valuation and budget profiles $\vect{v}, \vect{W}$.

\subsection{State of the Art and Our Contributions}
The PoA has been extensively analyzed across a variety of auction settings: including first-price and second-price auctions; scenarios involving only valuation-maximizing agents, only utility-maximizing agents, both types, or hybrid agents; environments constrained solely by budgets, solely by Return on Spend (RoS), or by both; as well as auction formats that incorporate reserve prices, additive boosts, and randomization. 
(We refer the reader to \Cref{sec:related} and the survey \cite[Section 4]{AggarwalBadanidiyuru24:Auto-biddingandauctions} for details.) 
Through those results (a part is summarized in \Cref{table:purePoA}), the PoA bound of 2 emerges as a barrier. 
This bound is tight for certain auction formats and can only be surpassed under special conditions, such as when 
there are exactly two agents \cite{LiawMehta23:Efficiencyofnon-truthful}, or when additional assumptions or external information (e.g., machine learning-based predictions) are introduced.
However, in general, this represents a strong barrier. This barrier is plausibly supported by \citet{LiawMehta23:Efficiencyofnon-truthful} in which 
they showed the PoA lower bound of 2 for a broad class of randomized mechanisms satisfying some natural properties. 
\citet{AggarwalBadanidiyuru24:Auto-biddingandauctions} raised the following intriguing question:

\begin{quote}
\centering
\em Design a mechanism that has a PoA of strictly less than 2 and furthermore compute the PoA as a function of $n$.
\end{quote}

\renewcommand{\arraystretch}{1.5}
\begin{table}[th]
  \small
  \centering
  \begin{tabular}{|c|c||c|c|}
    \hline
       \multicolumn{2}{|c||}{} & \multicolumn{2}{c|}{Agents' constraints}\\
       \cline{3-4}
     \multicolumn{2}{|c||}{}  & \textsc{RoS} & \textsc{Budget} + \textsc{RoS}\\
    \hline
    \hline
		\multirow{6}{*}{\rotatebox[origin=c]{90}{Agents objectives}}  &
    		  valuation maximizing    & \PoA(FPA) = 2 \cite{LiawMehta23:Efficiencyofnon-truthful} 
		 & \PoA(FPA) $= n$ \cite{LiawMehta24:Efficiencyofnon-truthful} \\
		 & & &  \PoA($q-$FPA) $= 2$ \cite{LiawMehta24:Efficiencyofnon-truthful} \\
		  & & &  \PoA(SPA) = 2 \cite{AggarwalBadanidiyuru19:Autobiddingwith}	\\
		 \cline{2-4}
		 & mix: both valuation     & \PoA(FPA) = 2.188 \cite{DengMao24:Efficiencyofthefirst-price}  
		 	&  \color{red}{\PoA(PM) = 2} \\ 
		 & and utility maximizing      &  &  \\
		 \cline{2-4}
		 & hybrid    &  & \color{red}{\PoA($m-$PM) = $1 + \frac{O(1)}{n-1}$} \\
    \hline
  \end{tabular}
  \vspace{0.3cm}
  \caption{Summary of results on the \PoA (of pure Nash equilibria) for different auction formats: first and second price auctions (FPA and SPA), 
  quasi-proportional FPA ($q-$FPA). Our contributions are in red w.r.t the proportional mechanism (PM) and its variant ($m$-PM).
	}
  \label{table:purePoA}
\end{table}

In this paper, we answer \emph{positively} this question in a general setting in which valuation functions are concave and there is a mix of 
utility-maximizing, valuation-maximizing agents, and also hybrid agents. Specifically, we consider the proportional mechanisms in which for each item $j$, agent $i$ receives 
a fraction of $\frac{b_{ij}}{b_{1j} + \ldots + b_{n,j}}$ of item $j$, and the payment of agent $i$ on item $j$ is its bid $b_{ij}$ in the (traditional) proportional mechanism
(known also as Kelly mechanism) or in general is a function of the bids of all agents for other variants of proportional mechanism.  
Our contributions are the following. 

\begin{enumerate}
	\item We establish the PoA upper bound of 2 for the proportional mechanism 
	 in settings that include both utility-maximizing and valuation-maximizing agents, and under 
	 both \textsc{Budget} and \textsc{RoS} constraints. Combing with the lower bound from \cite{CaragiannisVoudouris16:WelfareGuarantees} 
	(which holds even in the restricted case of only utility-maximizing agents and \textsc{Budget} constraint), this bound is \emph{tight} for the proportional mechanism.
	\item We introduce a variant of the proportional mechanism with an alternative payment scheme that achieves a PoA bound of 
	$1 + \frac{O(1)}{n-1}$ where $n \geq 2$ denotes the number of bidding agents. 
	This result not only surpasses the strong PoA barrier of 2, but also approaches full efficiency (i.e., PoA $\to 1$) as the number of agents increases. 
	As a notable consequence, our result circumvents previously established impossibility results in the context of autobidding \cite{LiawMehta23:Efficiencyofnon-truthful} 
	and resource allocation mechanisms \cite{CaragiannisVoudouris21:TheEfficiencyofResource}. 
	The results and further details are provided in \Cref{sec:efficient}.
\end{enumerate}
For the typical valuation functions $v_{i}(d_{i}) = \sum_{j=1}^{m} v_{ij} d_{ij}$ in the context of online advertising, the proportional mechanisms 
can be converted into randomized mechanisms.

Our approach is based on linear programming duality and optimality conditions in 
convex programming to analyze the PoA. 
At the level of agents, we characterize equilibrium structures using Karush-Kuhn-Tucker (KKT) conditions applied to convex programs of the agents' problems. 
At the level of the auctioneer, we employ a primal-dual approach to bound PoA. Specifically, we formulate a configuration integer program that represents the underlying optimization problem of the auctioneer. By relaxing the integer constraints, we derive the corresponding dual linear program, which, by weak duality, provides an upper bound on the optimal liquid welfare. Given an equilibrium, we leverage KKT-derived properties to construct feasible dual solutions, allowing us to bound PoA by analyzing the ratio between the primal objective (the liquid welfare at equilibrium) and the dual objective (an upper bound on the optimal liquid welfare).

%
%

\subsection{Further Related Works}	\label{sec:related}

\subparagraph{Auto-bidding Advertising and PoA.}  
The PoA in auto-bidding advertising has been actively studied recently in various auctions
\cite{AggarwalBadanidiyuru19:Autobiddingwith,LiawMehta23:Efficiencyofnon-truthful,Mehta22:Auctiondesign,DengMao24:Efficiencyofthefirst-price}.
A summary of PoA for pure Nash equilibria is given in \Cref{table:purePoA}. 
Beyond pure strategies, PoA has also been studied under more general equilibrium concepts, particularly mixed equilibria.
Notably, \citet{DengMao24:Efficiencyofthefirst-price} 
extended the PoA upper bound of 2 to mixed equilibria in first-price auctions involving valuation-maximizing agents subject to RoS constraints. Furthermore, they established a PoA bound of 2.188 in settings with both utility-maximizing and valuation-maximizing agents.

Randomized mechanisms have also been a focus of recent study. For the two-agent case (i.e., $n = 2$), \citet{Mehta22:Auctiondesign} showed a PoA of 1.9, which was subsequently improved to 1.8 in \cite{LiawMehta23:Efficiencyofnon-truthful}. In \cite{Mehta22:Auctiondesign}, the author also proved a PoA lower bound of 2  PoA for a broad class of randomized 
mechanisms that satisfy certain natural properties. 
Beyond equilibrium analysis, auto-bidding has been examined under no-regret learning dynamics with theoretical guarantees
\cite{FikiorisTardos23:Liquidwelfare,GaitondeLi23:BudgetPacing,LucierPattathil24:Autobidderswith}
and in machine-learning-augmented auctions that incorporate predictive models to improve bidding strategies \cite{BalseiroDeng21:Robustauction,DengMao24:Efficiencyofthefirst-price}



\subparagraph{Proportional mechanism in resource allocation.} 
The proportional mechanism, also known as the Kelly mechanism, has been studied by \citet{JohariTsitsiklis04:Efficiencyloss} in the context of classic social welfare
in which the tight PoA bound for pure Nash equilibria has been proven. 
Bridging the classic welfare and the liquid welfare, \citet{SyrgkanisTardos13:Composable-and-efficient} proved that the social welfare at equilibria of the proportional mechanism is at most
a constant factor away from the optimal liquid welfare.
In the single-item setting (i.e., $m=1$), 
building on a long line of works (\cite{CaragiannisVoudouris16:WelfareGuarantees,ChristodoulouSgouritsa16:OntheEfficiencyoftheProportional} among others),  \citet{CaragiannisVoudouris21:TheEfficiencyofResource} showed the tight PoA bound of 2 for pure Nash equilibria.
\citet{ChristodoulouSgouritsa16:OntheEfficiencyoftheProportional} considered the multiple-item setting and provided 
a PoA upper bound of 2.618 under more general equilibrium concepts, specifically coarse correlated equilibria. 
Proportional mechanisms with different payment schemes have been proposed. 
In settings with only utility-maximizing agents and no \textsc{Budget} or constraints,
full efficiency (i.e., PoA = 1) has been achieved using different payment designs 
\cite{MaheswaranBasar06:Efficientsignal,YangHajek07:VCG-Kellymechanisms,JohariTsitsiklis09:Efficiencyofscalar-parameterized}. 
However, \citet{CaragiannisVoudouris21:TheEfficiencyofResource} also established an impossibility result: no mechanism can achieve a PoA better than 2 for 
general concave valuation functions. In \Cref{sec:efficient}, we revisit this limitation and demonstrate that it can be overcome under an additional, natural assumption
(with a payment scheme different to the one in the standard proportional mechanism).


\section{Preliminaries}	\label{sec:preliminaries}

In this section, we define the class of proportional mechanisms considered in this paper and establish the existence of pure Nash equilibria within this framework. We also describe a transformation from a proportional mechanism—designed for divisible items—to a corresponding randomized mechanism applicable to indivisible items in autobidding settings. Finally, we formulate a suitable configuration linear program (LP) that serves for analyzing PoA in our paper.

\paragraph{Proportional Mechanisms.}
Recall that in the setting, there are $m$ different divisible items, and they are fractionally distributed among $n$ agents.
Each agent $i$ has a private monotone non-decreasing, concave, and differentiable valuation function 
$v_{i}: [0,1]^{m} \rightarrow \mathbb{R}_{\geq 0}$ so that $v_{i}(d_{i1}, \ldots, d_{ij})$ 
represents the value that agent $i$ achieves by receiving fractions $d_{ij}$ from items $j$. 
Additionally, agent $i$  has a private budget $W_{i}$, which restricts its payment. 
Each agent $1 \leq i \leq n$ strategically submits bids $b_{ij} \geq 0$ on item $1 \leq j \leq m$ subject to its budget $\sum_{j=1}^{m} b_{ij} \leq W_{i}$. 
After collecting all the bids, agent $i$ will receive a fraction of $\frac{b_{ij}}{b_{1,j} + \ldots + b_{n,j}}$ of item $j$, 
that is proportional to its submitted bid $b_{ij}$ on item $j$. In the traditional proportional mechanism (Kelly mechanism), the payment 
$p_{ij}$ of agent $i$ on item $j$ is its submitted bid $b_{ij}$. Different payment schemes have been proposed \cite{MaheswaranBasar06:Efficientsignal,YangHajek07:VCG-Kellymechanisms,JohariTsitsiklis09:Efficiencyofscalar-parameterized}. In this paper, we consider the following scheme inspired by the one in \cite{MaheswaranBasar06:Efficientsignal}.
\begin{equation}	\label{eq:payment}	
p_{ij} =  \biggl(\sum_{i' \neq i} b_{i'j} \biggr) 
	 	\cdot \int_{0}^{b_{ij}} \frac{g\bigl(  t  +  \sum_{i' \neq i} b_{i'j} \bigr)}{ \bigl( t  + \sum_{i' \neq i} b_{i'j} \bigr)^{2}}dt 
		+ h\biggl(\sum_{i' \neq i} b_{i'j} \biggr)
\end{equation}
where $g: \mathbb{R}_{\geq 0} \rightarrow \mathbb{R}_{\geq 0}$ is a non-decreasing function
and $h: \mathbb{R}_{\geq 0} \rightarrow \mathbb{R}_{\geq 0}$ is a function independent of $b_{ij}$. 
The traditional proportional mechanism corresponds to 
the case where $g(u) = \frac{u^{2}}{\sum_{i' \neq i} b_{i'j}}$ and $h \equiv 0$. 
Recall that the liquid welfare 
is $\sum_{i=1}^{n} \min \bigl \{ W_{i},  v_{i}\bigl( \frac{b_{i1}}{b_{i1} + \ldots + b_{i,m}}, \ldots, \frac{b_{i,m}}{b_{i1} + \ldots + b_{i,m}}\bigr) \bigr \}$.

\begin{proposition}	\label{prop:existence}
There always exists a pure Nash equilibrium in the proportional mechanisms with the payment schemes defined in \Cref{eq:payment}	
where $g(u) = \frac{1}{C} u^{2 + r}$ with $r \geq 0$. 
\end{proposition}
\begin{proof}
Recall that the problem of each agent $i$ with its parameter $0 \leq \rho_{i} \leq 1$ is the following 
$$
\max  v_{i}\bigl( d_{i} \bigr) - \rho_{i} \sum_{j=1}^{m} p_{ij} 
\quad 
\text{subject to:}
\quad
\sum_{j=1}^{m} p_{ij} \leq W_{i} ~ (\textsc{Budget}),  
\quad
\text{and} 
\quad
v_{i}\bigl( d_{i} \bigr)  \geq  \sum_{j=1}^{m} p_{ij} ~ (\textsc{RoS})
$$
where in the proportional mechanism $d_{ij} = \frac{b_{ij}}{b_{1,j} + \ldots + b_{n,j}}$. 
By the classic theorem of \cite{Debreu52:Asocialequilibrium,Fan52:Fixed-pointandminimax,Glicksberg52:Afurthergeneralization}, 
a pure Nash equilibrium exists if the three following conditions hold:
(i) compact and convex of the strategy set of each agent; 
(ii) continuity of each agent's objective w.r.t the strategy profile;
and (iii) quasiconcavity of the agent objective in the agent strategy set.
By the definition of payments and assumption on the functions $v_{i}$ (continuous differential concave)
and the set of constraints, the two first conditions are fulfilled. For the last condition, 
it is sufficient to verify whether the payment $p_{ij}$ is a convex function on $b_{ij}$.  
In particular, for function $g(u) = \frac{1}{C} u^{2 + r}$ with $r \geq 0$ and $C$ is a non-negative constant,
the second derivation of $p_{ij}$ w.r.t $b_{ij}$ is 
$$
\frac{1}{rC} \biggl(\sum_{i' \neq i} b_{i'j} \biggr) 
	 	\cdot \bigl( b_{ij}  + \sum_{i' \neq i} b_{i'j} \bigr)^{r - 1} \geq 0
$$
Therefore, the payment is convex and so the agent objective is concave.  
The theorem follows.
\end{proof}


\paragraph{Conversion to randomized algorithm.} 
Due to the allocation rule in proportional mechanisms, it is particularly convenient to transform such mechanisms 
into randomized mechanisms for autobidding advertising settings, especially when working with typical valuation functions $v_{i}(d_{i}) = \sum_{j=1}^{m} v_{ij} d_{ij}$.  
Specifically, given a proportional mechanism with payment $p_{ij}$, do the following.  
For each item $j$, given bids $b_{1j}, \ldots, b_{n,j}$, assign item $j$ to agent $i$ with probability 
$\frac{b_{ij}}{b_{1j} + \ldots + b_{n,j}}$ and charge agent $i$ a payment of  $\frac{b_{1j} + \ldots + b_{n,j}}{b_{ij}}p_{ij}$. 
Hence, the expected valuation of agent $i$ is $\sum_{j=1}^{m} \frac{b_{ij}}{b_{1j} + \ldots + b_{n,j}} v_{ij}$
and its expected payment is $\sum_{j=1}^{m} p_{ij}$ --- exactly the valuation and payment in the proportional mechanism. 
The \textsc{Budget} and \textsc{RoS} constraints now hold in expectation. 


\paragraph{Formulation.}

One of the key steps in our approach is to formulate a suitable configuration LP corresponding to the underlying optimization problem of the auctioneer 
to which we can apply the primal-dual analysis. Notice that we consider here the underlying optimization problem independent of any specific mechanism.
Fix an arbitrarily small constant $\epsilon > 0$. 
Let $\mathcal{D}(\epsilon) = \{k \cdot \epsilon : 0 \le k \le \frac{1}{\epsilon}\}$ be a discretized set of fractions. 
In the formulation, we will assume that an item can only be divided into fractions that belong to the set $\mathcal{D}(\epsilon)$. 
As $\epsilon$ can be chosen arbitrarily small and the valuations are continuous, this assumption will result in an arbitrarily small loss of the liquid welfare. 
For simplicity, as $\epsilon$ is fixed, we use the notation $\mathcal{D}$ instead of $\mathcal{D}(\epsilon)$ as long as it is clear from the context. 
A (fractional) \emph{assignment} (a solution) of items to agents is a set $S = \{(i, j, d_{ij}): 1 \leq i \leq n, 1 \leq j \leq m, d_{ij} \in \mathcal{D}\}$ such that 
$\sum_{i=1}^{n} d_{ij} \leq 1$ for all items $1 \leq j \leq m$. 
Intuitively, a tuple $(i,j,d_{ij})$ means that a fraction $d_{ij}$ of item $j$ is assigned to agent $i$.
Denote $\mathcal{S}$ a set of all (fractional) assignments.
Let $c_{S}$ be the total effective welfare of the assignment $S$, i.e., $c_{S} = \sum_{i=1}^{n}  \min\{W_i, v_i(d_{i,1}, \ldots, d_{i,m}) \}$ 
where 
$d_{ij}$ is the fraction of item $j$ that agent $i$ receives in the assignment $S$ for $1 \leq i \leq n$ and $1 \leq j \leq m$.
Let $z_{S}$ be a 0-1 variable indicating whether the assignment is chosen. Consider the following formulation.
\begin{minipage}[t]{0.57\textwidth}
	\begin{align*}
		&& \max  \sum_{S \in \mathcal{S}} &c_{S}\ z_{S} \\
		(\alpha_j) && \sum_{i=1}^{n} \sum_{d_{ij}} d_{ij} \sum_{S: (i,j,d_{ij}) \in S } z_{S} &= 1 & & \forall  j\\
		(\beta) && \sum_{S \in \mathcal{S}} z_{S}  &= 1	& & \\
		&& z_{S} &\in \{0,1\} & & \forall S \in \mathcal{S} 
	\end{align*}
\end{minipage}
\begin{minipage}[t]{0.4\textwidth}
	\begin{align*}
		\min \sum_{j=1}^{m} \alpha_j &+ \beta \\
		\sum_{(i,j,d_{ij}) \in S} d_{ij} \alpha_j + \beta &\geq c_{S}  & & \forall S \in \mathcal{S}\\
\end{align*}
\end{minipage}

In the formulation, the second constraint guarantees that an assignment must be selected. 
Moreover, by the integral constraint $z_{S} \in \{0,1\}$, there is exactly one selected assignment, i.e., all $z_{S}$'s but one are equal to 0.  
To understand the first constraint, observe that if one fixes an item $j$, an agent $i$, and an assigned fraction $d_{ij}$ of $j$ to agent $i$, 
the sum $\sum_{S: (i,j,d_{ij}) \in S } z_{S}$ equals to 1
iff $z_{S} = 1$ for some $S$ that contains $(i,j,d_{ij})$, or equals to 0 if $z_{S} = 0$ for all $S$ that contain $(i,j,d_{ij})$. 
In other words, the value of $\sum_{S: (i,j, d_{ij}) \in S } z_{S}$ (which equals either 0 or 1) indicates whether agent $i$ receives exactly the fraction of $d_{ij}$ of item $j$ in the solution. 
Therefore, summing up over all fractions in $\mathcal{D}$, the term $\sum_{d_{ij} \in \mathcal{D}} d_{ij} \cdot \sum_{S: (i,j, d_{ij}) \in S } z_{S}$ 
represents the fraction of item $j$ attributed to agent $i$. Consequently, the constraint $\sum_{i=1}^{n} \sum_{d_{ij} \in \mathcal{D}} d_{ij} \cdot \sum_{S: (i,j,d_{ij}) \in S } z_{S} = 1$
ensures the total fractions of item $j$ assigned to all agents sum up to 1. The objective is to maximize the total effective welfare. 
By relaxing the integrality of $z_{S}$, one can compute the dual LP on the right-hand side, representing an upper bound of the total liquid welfare. 

\paragraph{High level of our analysis strategy.} 
Given an equilibrium bid vector, we build a feasible dual solution of the above formulation. By the weak duality, the corresponding dual objective 
represents a lower bound of the optimal liquid welfare. Subsequently, we derive the PoA by bounding the dual objective to the liquid welfare of the equilibrium.  


\section{Efficiency of (Standard) Proportional Mechanism}		\label{sec:prop}
In this section, we settle the PoA of the standard proportional mechanism in the presence of both valuation-maximizing and utility-maximizing agents.

Let  $\vect{b}^{*}$ be an arbitrary pure Nash equilibrium and let $d^{*}_{ij}(\vect{b}^{*}) = b^{*}_{ij}/(b^{*}_{1,j} + \ldots + b^{*}_{n,j})$ be the fraction of item $j$ assigned to agent $i$ 
by the proportional mechanism in the equilibrium $\vect{b}^{*}$. When $\vect{b}^{*}$ is clear from the context, for simplicity, we drop the parameter $\vect{b}^{*}$
and simply use $d^{*}_{ij}$. 
Moreover, denote $B^{*}_{j} = \sum_{i=1}^{n} b^{*}_{ij}$ for every item $1 \leq j \leq m$.
In the mechanism, the payments equal the corresponding bids.
The problem of agent $i$ is the following. 
\begin{align*}
	 && \max_{b_{i1}, \ldots, b_{i,m}} ~ v_{i}\biggl(  \frac{b_{i1} }{b_{i1} +  \sum_{i' \neq i} b^{*}_{i',1}},  \ldots, \frac{b_{i,m}}{b_{i,m} +  \sum_{i'\neq i} b^{*}_{i',m}} & \biggr)  
	 - \rho_{i} \sum_{j=1}^{m} b^{*}_{ij}	&& \\
  (\lambda_{i}) &&	\sum_{j=1}^{m} b_{ij}	&\leq W_{i}  && (\textsc{Budget})\\
  (\mu_{i}) && v_{i}\biggl(  \frac{b_{i1} }{b_{i1} +  \sum_{i' \neq i} b^{*}_{i',1}},  \ldots, \frac{b_{i,m}}{b_{i,m} +  \sum_{i'\neq i} b^{*}_{i',m}} \biggr)  
	 &\geq  \sum_{j=1}^{m} b_{ij}  && (\textsc{RoS}) \\
	(\xi_{ij}) &&  b_{ij} &\geq 0	\qquad  \forall 1 \leq j \leq m 
\end{align*}
Recall that, in the above program, if $i$ is utility-maximizing agent, $\rho_{i} = 1$, 
and if $i$ is valuation-maximizing agent, $\rho_{i} = 0$.  

In the equilibrium $\vect{b}^{*}$, every agent $i$ maximizes its objective w.r.t the constraints. 
Given that the agent valuations are continuous and concave, 
the KKT condition (derivatives of the corresponding Lagrangian w.r.t variables $b_{ij}$) reads: for every $1 \leq j \leq m$, 
\begin{align*}
\frac{\sum_{i' \neq i} b^{*}_{i',j}}{\bigl( \sum_{i'=1}^{n} b^{*}_{i',j} \bigr)^{2}}
	&\cdot \frac{\partial v_{i}\bigl( d^{*}_{i1}, \ldots, d^{*}_{i,m} \bigr) }{\partial d_{ij}} 
	- \rho_{i} - \lambda_{i} + \xi_{ij} \\
	&+ \mu_{i} \biggl( \frac{\sum_{i' \neq i} b^{*}_{i',j}}{\bigl( \sum_{i'=1}^{n} b^{*}_{i',j} \bigr)^{2}}
	\cdot \frac{\partial v_{i}\bigl( d^{*}_{i1}, \ldots, d^{*}_{i,m} \bigr) }{\partial d_{ij}} - 1 \biggr)
	= 0
\end{align*}
Hence, 
\begin{align*}
(1 - d^{*}_{ij})
	\cdot \frac{\partial v_{i}\bigl( d^{*}_{i1}, \ldots, d^{*}_{i,m} \bigr) }{\partial d_{ij}} 
= \frac{\mu_{i} + \rho_{i} + \lambda_{i} - \xi_{ij}}{1 + \mu_{i}} B^{*}_{j}
	\qquad \forall 1 \leq j \leq m
\end{align*}

%
%

If $\sum_{j=1}^{m} b^{*}_{ij} < W_{i}$ (so $\lambda_{i} = 0$ by the slackness complementary condition) and $b^{*}_{ij} > 0$ (so $\xi_{ij} = 0$) then for every (either valuation-maximizing or utility-maximizing) agent $i$, it holds that
\begin{align}	\label{eq:KKT-multi}
(1 - d^{*}_{ij}) \frac{\partial v_{i}\bigl( d^{*}_{i1}, \ldots, d^{*}_{i,m} \bigr) }{\partial d_{ij}} =  \frac{\rho_{i} + \mu_{i} }{1 + \mu_{i}} B^{*}_{j} \leq B^{*}_{j}	\qquad \forall 1 \leq j \leq m
\end{align}
In particular, if $i$ is a utility-maximizing agent (i.e., $\rho_{i} = 1$)
and $\sum_{j=1}^{m} b^{*}_{ij} < W_{i}$ and $b^{*}_{ij} > 0$, then the inequality in \Cref{eq:KKT-multi} becomes equality.

%

\paragraph{Dual variable definition.}

Define the dual variables as the following: 
$\alpha_{j} =  \sum_{i=1}^{n} b^{*}_{ij} =  B^{*}_{j}$ 
and $\beta = \sum_{i=1}^{n} \beta_{i} $ where variables $\beta_{i}$ are defined in the following.

\begin{itemize}
	\item if $i$ is valuation-maximizing agent then $\beta_{i} = \min\{W_{i}, v_{i}(d^{*}_{i})\}$
	\item if $i$ is utility-maximizing agent then
		\begin{align*}
			\beta_{i} = \begin{cases}
					W_{i} & \text{if } v_{i}(d^{*}_{i1}, \ldots, d^{*}_{i,m}) \geq W_{i}, \\
					2 v_{i}(d^{*}_{i1}, \ldots, d^{*}_{im}) - \sum_{j=1}^{m} d^{*}_{ij} \cdot (1 - d^{*}_{ij}) 
						\cdot \frac{ \partial v_{i}(d^{*}_{i1}, \ldots, d^{*}_{i,m})}{\partial d_{ij}} & 		\text{otherwise}.
		\end{cases}
\end{align*}
\end{itemize}

\begin{lemma}
The dual variables defined above are feasible. 
\end{lemma}
\begin{proof}
Fix an arbitrary (fractional) assignment $S$ (consisting of tuples $(i,j,d_{ij})$). 
The dual constraint reads:
\begin{align*}
\sum_{(i,j,d_{ij}) \in S} d_{ij} \cdot \alpha_{j} + \beta &\geq c_{S} \\
\Leftrightarrow
\sum_{(i,j,d_{ij}) \in S} d_{ij} B^{*}_{j}  +  \sum_{i=1}^{n} \beta_{i}
	&\geq  \sum_{i=1}^{n}  \min\{W_i, v_i(d_{i,1}, \ldots, d_{i,m}) \}
\end{align*}

To prove the above inequality, it is sufficient to prove that for every fixed agent $1 \leq i \leq n$:
\begin{align}	\label{eq:multi-dual-constraint-i}
\sum_{j=1}^{m} d_{ij} B^{*}_{j}  + \beta_{i}
	\geq  \min\{W_i, v_i(d_{i,1}, \ldots, d_{i,m}) \}
\end{align}

Fix an agent $i$.
If $\beta_{i} = W_{i}$, then \Cref{eq:multi-dual-constraint-i} follows trivially. 
Until the end of the proof, assume that $\beta_{i} \neq W_{i}$. 
We consider different cases. 

\paragraph{Case 1: $i$ is valuation-maximizing agent.}
As $\beta_{i} \neq W_{i}$, it implies $\beta_{i} = v_{i}(d^{*}_{i})$ and $v_{i}(d^{*}_{i}) < W_{i}$. 
By \Cref{eq:KKT-multi}, we have
		\begin{align*}
		\sum_{j=1}^{m} d_{ij} B^{*}_{j} + \beta_{i} 
		&\geq 
		\sum_{j=1}^{m} d_{ij} (1 - d^{*}_{ij}) \frac{\partial v_{i}\bigl( d^{*}_{i1}, \ldots, d^{*}_{i,m} \bigr) }{\partial d_{ij}} + v_{i}(d^{*}_{i}) \\
		&\geq  \sum_{j=1}^{m} (d_{ij} - d^{*}_{ij}) \frac{\partial v_{i}\bigl( d^{*}_{i1}, \ldots, d^{*}_{i,m} \bigr) }{\partial d_{ij}} + v_{i}(d^{*}_{i}) \\
		&\geq v_{i}(d_{i1}, \ldots, d_{i,m})
		\geq \min \{ W_{i}, v_{i}(d_{i1}, \ldots, d_{i,m}) \}
		\end{align*}
where the second inequality holds since $d_{ij} \leq 1$ and the last inequality follows the concavity of $v_{i}$.

\paragraph{Case 2: $i$ is utility-maximizing agent.}
As $\beta_{i} \neq W_{i}$, it means that
$$
\beta_{i} = 2 v_{i}(d^{*}_{i1}, \ldots, d^{*}_{im}) - \sum_{j=1}^{m} d^{*}_{ij} (1 - d^{*}_{ij}) \frac{ \partial v_{i}(d^{*}_{i1}, \ldots, d^{*}_{i,m})}{\partial d_{ij}}
$$ (and also, $v_{i}(d^{*}_{i}) < W_{i}$).

By \Cref{eq:KKT-multi} and definition of $\beta_{i}$, we have
\begin{align*}
& \sum_{j=1}^{m} d_{ij} B^{*}_{j} + \beta_{i} \\
&\geq  \sum_{j=1}^{m} d_{ij} \bigl( 1 - d^{*}_{ij} \bigr) \cdot \frac{\partial v_{i}\bigl( d^{*}_{i1}, \ldots, d^{*}_{i,m} \bigr) }{\partial d_{ij}}  
	+ 2 v_{i}(d^{*}_{i1}, \ldots, d^{*}_{im}) - \sum_{j=1}^{m} d^{*}_{ij} (1 - d^{*}_{ij}) \frac{ \partial v_{i}(d^{*}_{i1}, \ldots, d^{*}_{i,m})}{\partial d_{ij}} \\
&=   v_{i}(d^{*}_{i1}, \ldots, d^{*}_{im}) + v_{i}(d^{*}_{i1}, \ldots, d^{*}_{im}) 
	+ \sum_{j=1}^{m} \bigl( d_{ij} - d^{*}_{ij} \bigr) \cdot \frac{\partial v_{i}\bigl( d^{*}_{i1}, \ldots, d^{*}_{i,m} \bigr) }{\partial d_{ij}} \\
	&\qquad \qquad \qquad \qquad
	- \sum_{j=1}^{m} d^{*}_{ij} \bigl( d_{ij} - d^{*}_{ij} \bigr) \cdot \frac{\partial v_{i}\bigl( d^{*}_{i1}, \ldots, d^{*}_{i,m} \bigr) }{\partial d_{ij}}  \\
&\geq  v_{i}(d^{*}_{i1}, \ldots, d^{*}_{im}) + v_{i}(d_{i1}, \ldots, d_{im}) 
	- \sum_{j=1}^{m} d^{*}_{ij} \bigl( d_{ij} - d^{*}_{ij} \bigr) \cdot \frac{\partial v_{i}\bigl( d^{*}_{i1}, \ldots, d^{*}_{i,m} \bigr) }{\partial d_{ij}}   \\
&\geq  v_{i}(d_{i1}, \ldots, d_{im}) + v_{i}(d^{*}_{i1}, \ldots, d^{*}_{im})
	- \sum_{j=1}^{m} d^{*}_{ij} \cdot \frac{\partial v_{i}\bigl( d^{*}_{i1}, \ldots, d^{*}_{i,m} \bigr) }{\partial d_{ij}}  \\
&\geq  v_{i}(d_{i1}, \ldots, d_{im}) + v_{i}(0, \ldots, 0) 
\geq v_{i}(d_{i1}, \ldots, d_{im}) 
\end{align*}
The first and third inequalities are due to the concavity of $v_{i}(\cdot)$. The second inequality holds since 
$d_{ij}, d^{*}_{ij} \leq 1$ and so $d_{ij} - d^{*}_{ij} \leq 1$. The last inequality follows $v_{i}(0, \ldots, 0) \geq 0$.

Combining all the cases above, the dual feasibility follows.
\end{proof}

\begin{theorem}
The PoA of pure Nash equilibria for liquid welfare in the proportional mechanism is at most 2.
\end{theorem}
\begin{proof}
We are bounding the liquid welfare of the (arbitrary) equilibrium $\vect{b}^{*}$ and the dual objective of the defined dual solution. 
The former is $\sum_{i=1}^{n} \min \{W_{i}, v_i(d^{*}_{i})\}$ whereas the latter is 
$\sum_{j=1}^{m} \alpha_{j} + \sum_{i=1}^{n} \beta_{i} 
=  \sum_{j=1}^{m} \sum_{i=1}^{n} b^{*}_{ij} + \sum_{i=1}^{n} \beta_{i} 
=  \sum_{i=1}^{n} (\sum_{j=1}^{m} b^{*}_{ij} + \beta_{i}) $. 
Again, it is sufficient to prove that for each $1 \leq i \leq n$,  
$$
\sum_{j=1}^{m} b^{*}_{ij} + \beta_{i} 
\leq 2 \min \{W_{i}, v_i(d^{*}_{i})\}.
$$ 

If $i$ is valuation-maximizing agent, $\sum_{j=1}^{m} b^{*}_{ij} \leq \min\{W_{i}, v_{i}(d^{*}_{i})\}$.
By definition of dual variables, $\beta_{i} = \min\{W_{i}, v_{i}(d^{*}_{i})\}$. Therefore, the above inequality follows.
In the following, consider a utility-maximizing agent $i$.   


\paragraph{Case 1:  $\beta_{i} = W_{i}$.} 
This case also means that $v_i(d^{*}_{i}) \geq W_{i}$.
Therefore,
\begin{align*}
 \sum_{j=1}^{m} b^{*}_{ij} + \beta_{i} 
= \sum_{j=1}^{m} b^{*}_{ij} + W_{i} 
\leq 2  W_{i} 
= 2  \min \{ W_{i}, v_i(d^{*}_{i}) \} \bigr].
\end{align*}

\paragraph{Case 2: $\beta_{i} \neq W_{i}$.}
By definition of $\beta_{i}$, the case assumption implies $v_{i}(d^{*}_{i}) \leq W_{i}$ and $\sum_{j=1}^{m} b^{*}_{ij} < W_{i}$. 
(Otherwise, if $\sum_{j=1}^{m} b^{*}_{ij} = W_{i}$ then by the non-negativity of agent $i$'s utility, $v_{i}(d^{*}_{i}) \geq \sum_{j=1}^{m} b^{*}_{ij} = W_{i}$ 
and that implies $\beta_{i} = W_{i}$ --- contradicting to the case assumption.) 
Hence, 
\begin{align*}
 \sum_{j=1}^{m} b^{*}_{ij} + \beta_{i} 
&=   \sum_{j=1}^{m} b^{*}_{ij} + 2 v_{i}(d^{*}_{i1}, \ldots, d^{*}_{im}) - \sum_{j=1}^{m} d^{*}_{ij} (1 - d^{*}_{ij}) \cdot \frac{ \partial v_{i}(d^{*}_{i1}, \ldots, d^{*}_{i,m})}{\partial d_{ij}}  \\
&=   \sum_{j: b^{*}_{ij} \neq 0} b^{*}_{ij} + 2 v_{i}(d^{*}_{i1}, \ldots, d^{*}_{im}) - \sum_{j: b^{*}_{ij} \neq 0} d^{*}_{ij} B^{*}_{j}  \\
&=  \sum_{j: b^{*}_{ij} \neq 0} b^{*}_{ij} + 2 v_{i}(d^{*}_{i1}, \ldots, d^{*}_{im}) - \sum_{j: b^{*}_{ij} \neq 0} b^{*}_{ij} \\
&=   2 v_{i}(d^{*}_{i1}, \ldots, d^{*}_{im})  
= 2  \min \{W_{i}, v_i(d^{*}_{i})\}
\end{align*}
where the second equality is due to \Cref{eq:KKT-multi} with equality for utility-maximizing agents. The theorem follows.
\end{proof}


\section{Mechanism with asymptotically full efficiency}	\label{sec:efficient}
In this section, we study a proportional mechanism with a different payment scheme in the setting with general hybrid agents (each has its own parameter $\rho_{i}$).   
Fix an arbitrarily small constant $\epsilon \geq \frac{1}{n-1}$. 
The payment is specifically in the form of \Cref{eq:payment} in which function
$g(u) = u^{1+(n-1)\epsilon}$ and $h(u) = \frac{g(u)}{(n-1)\epsilon}$. This choice of function $g$ will be clear later.
The payment is explicitly given in \Cref{eq:payment-opt}.
To guarantee the existence of pure Nash equilibria, by \Cref{prop:existence}, one needs to choose $\epsilon$ such that 
$1+(n-1)\epsilon \geq 2$, that explains the condition $\epsilon \geq \frac{1}{n-1}$. 

Let  $\vect{b}^{*}$ be an arbitrary pure Nash equilibrium and let $d^{*}_{ij}(\vect{b}^{*})$ be the fraction of item $j$ assigned to agent $i$ 
by the proportional mechanism given $\vect{b}^{*}$. When $\vect{b}^{*}$ is clear from the context, for simplicity, we drop the parameter $\vect{b}^{*}$
and simply use $d^{*}_{ij}$. 
Moreover, denote $B^{*}_{j} = \sum_{i=1}^{n} b^{*}_{ij}$ for every item $1 \leq j \leq m$.
We recall the problem of (hybrid) agent $i$ with a new payment scheme. 
\begin{align}	\label{eq:payment-opt}
	 && &\max_{b_{i1}, \ldots, b_{i,m}} ~  v_{i}\biggl(  \frac{b_{i1} }{b_{i1} +  \sum_{i' \neq i} b^{*}_{i',1}},  \ldots, \frac{b_{i,m}}{b_{i,m} +  \sum_{i'\neq i} b^{*}_{i',m}} \biggr) 
	 - \rho_{i} \cdot \sum_{j=1}^{m} p_{ij} 	 \notag \\
	&& &p_{ij} = \biggl(\sum_{i' \neq i} b^{*}_{i'j} \biggr) 
	 	\cdot \biggl[ \int_{0}^{b_{ij}} \frac{g\bigl(  t  +  \sum_{i' \neq i} b^{*}_{i'j} \bigr)}{ \bigl( t  + \sum_{i' \neq i} b^{*}_{i'j} \bigr)^{2}}dt 
					+ \frac{g\bigl(\sum_{i' \neq i} b^{*}_{i'j} \bigr)}{ (n-1)\epsilon \cdot \sum_{i' \neq i} b^{*}_{i'j}} \biggr] 
						&&  \forall 1 \leq j \leq m \\
  	(\lambda_{i}) &&	&\sum_{j=1}^{m} p_{ij}	\leq W_{i}  && (\textsc{Budget})		\notag \\
	(\mu_{i}) &&   &v_{i}\biggl(  \frac{b_{i1} }{b_{i1} +  \sum_{i' \neq i} b^{*}_{i',1}},  \ldots, \frac{b_{i,m}}{b_{i,m} +  \sum_{i'\neq i} b^{*}_{i',m}} \biggr) 
	 	\geq \sum_{j=1}^{m} p_{ij}  && (\textsc{RoS})  	\notag \\
	(\xi_{ij}) && 	& b_{ij} \geq 0	&&  \forall 1 \leq j \leq m 		\notag 
\end{align}

Given that the agent valuations are continuous, differentiable, and concave, 
the KKT condition (derivatives of the corresponding Lagrangian function w.r.t variables $b_{ij}$) reads: for every $1 \leq i \leq n$, $1 \leq j \leq m$,  
\begin{align}	\label{eq:KKT-opt-pre}
\frac{ \sum_{i' \neq i} b^{*}_{i',j}}{\bigl(  \sum_{i'=1}^{n} b^{*}_{i',j} \bigr)^{2}}
	&\cdot \biggl[  \frac{\partial v_{i}\bigl( d^{*}_{i1}, \ldots, d^{*}_{i,m} \bigr) }{\partial d_{ij}} 
	- \rho_{i} \cdot g\biggl( \sum_{i'=1}^{n} b^{*}_{i'j} \biggr) \biggr]
	+ \xi_{ij} 
	- \lambda_{i} \frac{\sum_{i' \neq i} b^{*}_{i',j}}{\bigl(  \sum_{i'=1}^{n} b^{*}_{i',j} \bigr)^{2}}
		\cdot g\biggl( \sum_{i'=1}^{n} b^{*}_{i'j} \biggr)	\notag \\
	&+ \mu_{i} \frac{\sum_{i' \neq i} b^{*}_{i',j}}{\bigl(  \sum_{i'=1}^{n} b^{*}_{i',j} \bigr)^{2}}
		\cdot \biggl[  \frac{\partial v_{i}\bigl( d^{*}_{i1}, \ldots, d^{*}_{i,m} \bigr) }{\partial d_{ij}} 
		- g\biggl( \sum_{i'=1}^{n} b^{*}_{i'j} \biggr) \biggr]
	= 0
\end{align}
Consequently, if $\sum_{j=1}^{m} p^{*}_{ij} < W_{i}$ (so $\lambda_{i} = 0$ by the slackness complementary condition), then, as $\xi_{ij} \geq 0$, it holds that:
\begin{align}	\label{eq:KKT-opt}
	\frac{\partial v_{i}\bigl( d^{*}_{i1}, \ldots, d^{*}_{i,m} \bigr) }{\partial d_{ij}} 
		\leq  \frac{\rho_{i} + \mu_{i}}{1 + \mu_{i}} g\biggl( \sum_{i'=1}^{n} b^{*}_{i'j} \biggr)
		\leq  g\biggl( \sum_{i'=1}^{n} b^{*}_{i'j} \biggr)
		\qquad \forall i, j 
\end{align}

\subparagraph{Dual variables.}

Define the dual variables as the following: 
$\alpha_{j} =  g\bigl( B^{*}_{j}\bigr)$ 
and $\beta = \sum_{i=1}^{n} \beta_{i} $ where
$$
\beta_{i} = 
	\begin{cases}
		v_{i}(d^{*}_{i}) - \sum_{j=1}^{m} d^{*}_{ij} \frac{\partial v_{i}\bigl( d^{*}_{i1}, \ldots, d^{*}_{i,m} \bigr) }{\partial d_{ij}} & \text{ if } v_{i}(d^{*}_{i}) < W_{i}, \\
		W_{i}	& \text{ otherwise}. 
	\end{cases}
$$

\begin{lemma}	\label{lem:price-g}
For every item $1 \leq j \leq m$, it holds that $\epsilon \cdot \sum_{i=1}^{n} p^{*}_{ij} = g \bigl(  \sum_{i = 1}^{n} b^{*}_{ij} \bigr)$.
\end{lemma}
\begin{proof}
Fix an arbitrary item $1 \leq j \leq m$. For simplicity, as $j$ is fixed, denote (only in this proof) $B^{*}_{-i} = \sum_{i' \neq i} b^{*}_{i'j}$ and $B^{*} = \sum_{i'=1}^{n} b^{*}_{i'j}$
without the subindex $j$. We have:
\begin{align*}
p^{*}_{ij} &=  \biggl(\sum_{i' \neq i} b^{*}_{i'j} \biggr) 
	 	\cdot \biggl[ \int_{0}^{b^{*}_{ij}} \frac{g\bigl(  t  +  \sum_{i' \neq i} b^{*}_{i'j} \bigr)}{ \bigl( t  + \sum_{i' \neq i} b^{*}_{i'j} \bigr)^{2}}dt 
					+ \frac{g\bigl(\sum_{i' \neq i} b^{*}_{i'j} \bigr)}{ (n-1)\epsilon \cdot \sum_{i' \neq i} b^{*}_{i'j}} \biggr] \\
	&= B^{*}_{-i} \cdot \biggl[ \int_{0}^{b^{*}_{ij}} \frac{g\bigl(  t  +  B^{*}_{-i} \bigr)}{ \bigl( t  + B^{*}_{-i} \bigr)^{2}}dt 
					+ \frac{g\bigl(B^{*}_{-i} \bigr)}{ (n-1)\epsilon \cdot B^{*}_{-i} } \biggr]
\end{align*}

By the choice of $g(u)  = e^{(1 + (n-1)\epsilon)\ln u} = u^{1+(n-1) \epsilon}$, we have 
\begin{align}	\label{eq:int-g}
\int_{0}^{b^{*}_{ij}} \frac{g\bigl(  t  +  B^{*}_{-i} \bigr)}{ \bigl( t  + B^{*}_{-i} \bigr)^{2}}dt 
&= \int_{0}^{b^{*}_{ij}} \frac{1}{ \bigl( t  + B^{*}_{-i} \bigr)^{1- (n-1)\epsilon}}dt 
= \frac{1}{(n-1)\epsilon} \biggl[ (B^{*})^{(n-1)\epsilon} - (B^{*}_{-i})^{(n-1)\epsilon} \biggr] 	\notag \\
&= \frac{1}{(n-1)\epsilon} \biggl[ \frac{(B^{*})^{1 + (n-1)\epsilon}}{B^{*}} - \frac{(B^{*}_{-i})^{1 + (n-1)\epsilon}}{B^{*}_{-i}} \biggr] \notag \\
&= \frac{1}{(n-1)\epsilon} \biggl[ \frac{g(B^{*})}{B^{*}} - \frac{g(B^{*}_{-i})}{B^{*}_{-i}} \biggr]
\end{align}
Therefore, 
\begin{align*}
\sum_{i=1}^{n} p^{*}_{ij} 
	&= \sum_{i=1}^{n}  \frac{1}{(n-1)\epsilon} B^{*}_{-i} \frac{g(B^{*})}{B^{*}}
	= \frac{1}{(n-1)\epsilon}  \frac{g(B^{*})}{B^{*}} \sum_{i=1}^{n} B^{*}_{-i}  =  \frac{1}{\epsilon} g(B^{*})
\end{align*}
where $\sum_{i=1}^{n} B^{*}_{-i} = (n-1)B^{*}$.
\end{proof}

\subparagraph{Remark.} We choose $g(u)  = e^{(1 + (n-1)\epsilon)\ln u} = u^{1+(n-1) \epsilon}$ to fullfil \Cref{lem:price-g}. 
In particular,  $g(u)  = e^{(1 + (n-1)\epsilon)\ln u}$ is a solution of the following differential equation:
$$
\frac{d}{du} \biggl( \frac{g(u)}{(n-1)\epsilon \cdot u} \biggr) = \frac{g(u)}{u^{2}}
$$ 
This equation is needed in the proof of \Cref{lem:price-g}, specifically, \Cref{eq:int-g}. 

\begin{lemma}
The dual variables defined above are feasible. 
\end{lemma}
\begin{proof}
Fix an arbitrary (fractional) assignment $S$ (consisting of tuples $(i,j,d_{ij})$). 
The dual constraint reads:
\begin{align*}
\sum_{(i,j,d_{ij}) \in S} d_{ij} \cdot \alpha_{j} + \beta &\geq c_{S} \\
\Leftrightarrow
\sum_{(i,j,d_{ij}) \in S} d_{ij} g(B^{*}_{j})  +  \sum_{i=1}^{n} \beta_{i}
	&\geq  \sum_{i=1}^{n}  \min \bigl \{ W_{i},  v_i(d_{i,1}, \ldots, d_{i,m}) \bigr \}
\end{align*}

To prove the above inequality, it is sufficient to prove that for every fixed agent $1 \leq i \leq n$:
\begin{align}	\label{eq:dual-constraint-i}
\sum_{j=1}^{m} d_{ij} g(B^{*}_{j})  + \beta_{i}
	\geq  \min \bigl \{ W_{i},  v_i(d_{i,1}, \ldots, d_{i,m}) \bigr \}
\end{align}

Fix an agent $i$. If $\beta_{i} = W_{i}$ then \Cref{eq:dual-constraint-i} holds trivially. Assume that $\beta_{i} < W_{i}$. By definition of $\beta_{i}$
and the \textsc{RoS} constraint, it implies 
that $\sum_{j=1}^{m}p^{*}_{ij} \leq v_{i}(d^{*}_{i}) < W_{i}$. 
By \Cref{eq:KKT-opt}, we have
\begin{align*}
		\sum_{j=1}^{m} d_{ij} g\bigl(  B^{*}_{j} \bigr) + \beta_{i} 
		&\geq 
		\sum_{j=1}^{m} d_{ij} \frac{\partial v_{i}\bigl( d^{*}_{i1}, \ldots, d^{*}_{i,m} \bigr) }{\partial d_{ij}} 
			+ v_{i}(d^{*}_{i}) - \sum_{j=1}^{m} d^{*}_{ij} \frac{\partial v_{i}\bigl( d^{*}_{i1}, \ldots, d^{*}_{i,m} \bigr) }{\partial d_{ij}}\\
		&=  \sum_{j=1}^{m} (d_{ij} - d^{*}_{ij}) \frac{\partial v_{i}\bigl( d^{*}_{i1}, \ldots, d^{*}_{i,m} \bigr) }{\partial d_{ij}} + v_{i}(d^{*}_{i}) \\
		&\geq v_{i}(d_{i1}, \ldots, d_{i,m})
\end{align*}
where the last inequality is due to the concavity of $v_{i}$. The lemma follows.
\end{proof}

\begin{theorem}	\label{thm:opt}
The price of anarchy of the auction is at most $1 + \epsilon$ for any $\epsilon \geq \frac{1}{n-1}$.
\end{theorem}
\begin{proof}
We bound the dual objective by the equilibrium's liquid welfare. Note that, by the definition, $\beta_{i} \leq \min \bigl\{ W_{i}, v_{i}(d^{*}_{i}) \bigr\}$.
We have:
\begin{align*}
\sum_{j=1}^{m} \alpha_{j} + \sum_{i=1}^{n} \beta_{i}
&\leq \sum_{j=1}^{m} g\bigl(  B^{*}_{j} \bigr) + \sum_{i=1}^{n} \min \bigl\{ W_{i}, v_{i}(d^{*}_{i}) \bigr\} \\
%
%
&\leq  \sum_{i=1}^{n} \biggl( \epsilon \sum_{j=1}^{m} p^{*}_{ij} +  \min \bigl\{ W_{i}, v_{i}(d^{*}_{i}) \bigr\} \biggr)
\leq (1 + \epsilon) \sum_{i=1}^{n} \min \bigl\{ W_{i}, v_{i}(d^{*}_{i}) \bigr\} 
\end{align*}
where the second equality holds by \Cref{lem:price-g}, and the last inequality follows the \textsc{RoS} and \textsc{Budget} constraints, 
i.e., the total payment of agent $i$ is less than its valuation and its budget.
By the weak duality, the theorem follows.
\end{proof}

\subparagraph{Discussion.} As a notable consequence of \Cref{thm:opt}, our result circumvents previously established impossibility results 
in the context of autobidding \cite{LiawMehta23:Efficiencyofnon-truthful} 
and resource allocation mechanisms \cite{CaragiannisVoudouris21:TheEfficiencyofResource}. 
A crucial condition in their impossibility result is that an agent’s payment must not exceed their submitted bid. In our mechanism, however, the payment 
	$p_{ij}$ (defined in \Cref{eq:payment-opt}) can exceed $b_{ij}$ (but always satisfies \textsc{Budget} and \textsc{RoS} constraints),  
	thus allowing us to bypass this limitation and achieve the improved PoA stated in
	\Cref{thm:opt}. 

%

\subsection{A modified mechanism}

%

While the only required constraints in auto-bidding setting are \textsc{Budget} and \textsc{RoS}, an additional desirable property in classic mechanism design is that payments do not exceed submitted bids. 
We propose a modified mechanism that upholds this property under the following mild and realistic assumption:
%
\begin{quote}
	\textbf{Assumption:} there exist a publicly known parameter $W \geq 1$ such that  all agents’ bids are bounded by $W$. 
\end{quote}

Note that in our setting, given a (finite) budget bound $W_{i}$, it follows that in proportional mechanisms, whether using the traditional payment rule or the one defined in \Cref{eq:payment-opt}, every agent $i$
will naturally submit bounded bids (depending on its budget $W_{i}$).
One can consider parameter $W$ in the assumption as a sufficiently large upper bound of all $W_{i}$.

\subparagraph{Mechanism.} 
Fix an arbitrary constant $\epsilon \geq \frac{1}{n-1}$. Collect all submitted bids $\tilde{b}_{ij}$'s and compute modified bids 
$b_{ij} := \frac{1}{nW} \max\{\tilde{b}_{ij} - \frac{1}{(n-1)\epsilon}, 0\}$ for every item $1\leq j \leq m$ 
and agent $1 \leq i \leq n$.  
Run the proportional mechanism with the payment scheme \Cref{eq:payment-opt} on the
modified bid vectors $b_{ij}$'s. 

\vspace{0,3cm}

At a high level, the value of $\frac{1}{(n-1)\epsilon}$ 
acts as a threshold to filter bids, serving a role intuitively similar (but not the same) to that of a reserve price. 
The mechanism’s payment is always determined by \Cref{eq:payment-opt} applied to the modified bids.
In this formulation, all modified bids are uniformly scaled by a factor of 
$\frac{1}{nW}$ (ensuring that $b_{ij} \leq \frac{1}{n}$). 
The scaling and the value of $\frac{1}{(n-1)\epsilon}$ are chosen for the purpose of \Cref{lem:price<bid}, and are used solely for payment computation. 
The allocation in the proportional mechanism can be computed without this scaling factor.

\begin{lemma}		\label{lem:price<bid}
By the payment scheme defined by \Cref{eq:payment-opt}, it always holds that the payments do not exceed submitted bids, i.e., $p_{ij} \leq \tilde{b}_{ij}$ for all $i,j$.
\end{lemma}
\begin{proof}
Given the bids $\tilde{b}_{ij}$'s and the corresponding modified bids $b_{ij}$'s, let $B_{j} = \sum_{i=1}^{n} b_{ij}$. 
By \Cref{eq:payment-opt} and the choice of function $g$, the payment $p_{ij}$ is:
\begin{align*}
p_{ij} &= \biggl(\sum_{i' \neq i} b_{i'j} \biggr) 
	 	\cdot \biggl[ \int_{0}^{b_{ij}} \frac{g\bigl(  t  +  \sum_{i' \neq i} b_{i'j} \bigr)}{ \bigl( t  + \sum_{i' \neq i} b_{i'j} \bigr)^{2}}dt 
					+ \frac{g\bigl(\sum_{i' \neq i} b_{i'j} \bigr)}{ (n-1)\epsilon \cdot \sum_{i' \neq i} b_{i'j}} \biggr] \\
&\leq B_{j}^{(n-1)\epsilon} b_{ij} + \frac{1}{(n-1)\epsilon} B_{j}^{1 + (n-1)\epsilon}
\leq b_{ij} + \frac{1}{(n-1)\epsilon} \leq \tilde{b}_{ij}
\end{align*}
since $B_{j} = \sum_{i=1}^{n} b_{ij} \leq 1$. 
\end{proof}

\begin{theorem}	\label{thm:opt-modified}
For any constant $\epsilon \geq 1/(n-1)$, let $\vect{\tilde{b}}$ be a pure Nash equilibrium. 
Then, the liquid welfare of equilibrium $\vect{\tilde{b}}$ is at most 
$1 + \epsilon$ times that of the optimal solution. 
\end{theorem}
\begin{proof}
The proof of this theorem is literally the same as the one in \Cref{thm:opt}. 
The only subtle point is to verify that for each item $j$, there exists a modified bid $b_{ij} >0$, i.e., 
there exists a bid $\tilde{b}_{ij} > \frac{1}{(n-1)\epsilon}$. 
It is required in the analysis of KKT conditions leading to \Cref{eq:KKT-opt}. 
Specifically, it is necessary to avoid the division by 0 
in \Cref{eq:KKT-opt-pre}.
In particular, one can observe that in an equilibrium $\vect{\tilde{b}}$, if 
there exists an item $j$ such that $\tilde{b}_{i'j} \leq \frac{1}{(n-1)\epsilon}$ for all $1 \leq i' \leq n$, then an arbitrary 
agent $i$ has an incentive to submit a bid $\tilde{b}_{ij} > \frac{1}{(n-1)\epsilon}$. 
By doing that, agent $i$ get the entire item $j$ with the payment of 0 (by \Cref{eq:payment-opt}) and strictly increase its objective 
(for whatever the value of $\rho_{i}$, the valuation of agent $i$ strictly increases and the payment is equal to 0).  
Hence, in an equilibrium, for every item $j$, there exists a modified bid $b_{ij} >0$. By the same steps in \Cref{thm:opt}, the theorem follows. 
\end{proof}


\section{Conclusion}


In this paper, we established the tight PoA bound of 2 for the standard proportional mechanism in auto-bidding systems where agents have hybrid objectives under both \textsc{Budget} and \textsc{RoS} constraints. Furthermore, we introduced a variant with a modified payment scheme that asymptotically achieves full efficiency as the number of agents grows. This result surpasses the previously established PoA barrier of 2 and circumvents known impossibility results in both autobidding and resource allocation settings.

As observed, the payment scheme defined by \Cref{eq:payment-opt} is more intricate than those used in first-price and second-price auctions, as well as in the standard proportional mechanism. From the perspective of classical auction design, it would be interesting to develop a simpler payment rule while preserving comparable PoA guarantees. From the standpoint of auto-bidding environments, however, advertisers are primarily concerned with achieving their own objectives and typically delegate complex bidding decisions to automated agents that translate these goals into strategic bids subject to relevant constraints. This delegation provides flexibility in payment design beyond the traditional desiderata imposed on payment rules.

Finally, our methodology leverages duality and the KKT conditions from linear and convex programming. This approach may have broader applications for establishing PoA bounds in auto-bidding and other strategic settings.




\bibliographystyle{plainnat}
\bibliography{references}




\end{document}